\documentclass[reqno]{amsart}
\usepackage[foot]{amsaddr}

\usepackage[utf8]{inputenc}

\usepackage{svg}
\usepackage[hidelinks]{hyperref}
\usepackage{url}

\usepackage{amssymb}
\usepackage{amscd}
\usepackage{amsthm}
\usepackage{setspace}
\usepackage{enumerate}
\usepackage{graphicx}
\usepackage{mathtools}
\usepackage{tcolorbox} 
\usepackage{subcaption} 
\usepackage{siunitx}
\usepackage{tikz-cd}
\usepackage{bm,blkarray}
\numberwithin{equation}{section}
\usepackage{subcaption} 
\usepackage{tikz}
\usepackage{pgfplots}
\usepackage{thmtools}
\pgfplotsset{compat=1.18}

\usepackage{mathtools}
\usepackage[tableposition=top]{caption}
\usepackage{booktabs,dcolumn}

\renewcommand\P{\mathbb{P}}
\newcommand\E{\mathbb{E}}

\newcommand\R{\mathbb{R}}

\newcommand\bg{\bm{g}}
\newcommand\bw{\bm{w}}
\newcommand\hg{\hat{\bm g}}

\newcommand\rem{\tilde{O}\left(n^{-3/2}\right)}
\newcommand\bu{\bm{u}}
\newcommand\hu{\hat{\bm u}}
\newcommand\bv{\bm{v}}
\newcommand\hv{\hat{\bm v}}

\newcommand\eps{\varepsilon}
\newcommand\Tr{{\operatorname{Tr}}}
\newcommand\on[1]{{\mathbf 1}_{\{#1\}}}
\newcommand\proj{{\operatorname{P}}}

\newcommand{\quash}[1]{}

\newcommand{\diag}{\mathop {\rm diag}}

\newcommand{\av}{\mathrm{avg}}
\newcommand{\std}{\mathrm{std}}
\newcommand{\emp}{\mathrm{emp}}
\newcommand{\Nemp}{\mathrm{emp}_{\mathcal{N}}}
\newcommand{\Var}{\mathrm{Var}}
\newcommand{\Cov}{\mathrm{Cov}}

\title[Nodal Count for Orthogonally Invariant Ensembles]{Nodal Count for Orthogonally Invariant Ensembles}
\author{Lior Alon}
\address{Department of Mathematics, Massachusetts Institute of Technology, Cambridge, MA, 02139 USA.}
\email{lioralon@mit.edu}
\author{Dan Mikulincer}
\address{Department of Mathematics, University of Washington, Seattle, WA 98109 USA.}
\email{danmiku@uw.edu}
\author{John Urschel}
\email{urschel@mit.edu}

\newtheorem{theorem}{Theorem}[section]
\newtheorem{definition}[theorem]{Definition}
\newtheorem{notations}[theorem]{Notations}
\newtheorem{lemma}[theorem]{Lemma}
\newtheorem{conjecture}[theorem]{Conjecture}
\newtheorem{proposition}[theorem]{Proposition}

\newtheorem{corollary}[theorem]{Corollary}

\begin{document}

\begin{abstract}
We investigate the nodal count of eigenvectors of random matrices interpreted as operators on signed complete graphs. Our focus is on orthogonally invariant ensembles, with particular attention to the Gaussian Orthogonal Ensemble (GOE). We establish that, as the matrix size tends to infinity, the distribution of nodal counts converges to the same limiting law as the eigenvalue distribution. In the GOE case, this limit is the semicircle law. This result refutes a conjecture, motivated by quantum chaos and quantum graphs, which predicted Gaussian behavior of the nodal count.
\end{abstract}

\maketitle

\section{Introduction}
The purpose of this paper is to investigate the asymptotic nodal count sequence (size of the nodal sets of eigenvectors) of a random real symmetric $n\times n$ matrix as $n\to\infty$. The nodal count sequence of a matrix $A$ is $\phi(A)=\left(\phi(A,k)\right)_{k=1}^n$, with  
\[\phi(A,k):=\left|\{i<j,\ :\ A_{ij}\bm\varphi^{(k)}_{i}\bm\varphi^{(k)}_{j}>0\}\right|,\]
  where $\bm\varphi^{(k)}$ is the eigenvector of the  $k$-th smallest eigenvalue of $A$.
We consider ensembles of matrices that are invariant under orthogonal transformations, most notably, the Gaussian orthogonal ensemble (GOE). Our result for GOE matrices disproves a recent conjecture about the nodal count for randomly signed graphs \cite{alon2023morse}. While our focus is on the discrete setting—namely, nodal counts for matrices—our motivation stems from foundational results on the measure of nodal sets of eigenfunctions on manifolds, as well as the significant progress over the past two decades on the study of nodal sets of random waves.

The study of nodal patterns dates back to the 17th century when it was popularized by Ernst Chladni, who revealed intricate sand patterns on vibrating plates. Over time, mathematicians developed a rigorous theory to understand the interplay between order and chaos in these patterns, paving the way for spectral geometry. The first rigorous study of nodal sets was by Sturm in 1836 \cite{Sturm1836}, who showed that the $k$-th eigenfunction of a Sturm–Liouville operator on an interval has exactly $k-1$ zeros. In higher dimensions, nodal sets—the zero sets of eigenfunctions—exhibit much greater complexity. This naturally leads to questions about their geometric properties and how these depend on the underlying domain or manifold. 

A fundamental question in spectral geometry regarding the size of nodal sets was asked by S.T. Yau \cite{Yau1982}. For a smooth compact manifold $M$, let  $\lambda_{k}$ denote the $k$-th smallest eigenvalue of the Laplacian $\Delta_M$ and $\phi(\Delta_M,k)$ denote the measure of the nodal set of its eigenfunction. Yau's conjecture states that $c_1\sqrt{\lambda_{k}}<\phi(\Delta_M,k)<c_2\sqrt{\lambda_{k}}$ for some positive constants $c_1,c_2$.  The proof for real analytic manifolds was given by Donnelly and Fefferman \cite{donnelly1988nodal}. For smooth manifolds, Logunov and Malinnikova proved the lower bound and gave a polynomial upper bound \cite{Logunov2018Lower,Logunov2018Upper,LogunovMalinnikova2018}. 

To go beyond asymptotic bounds, additional structure is required. For instance, much more can be computed explicitly in the presence of Gaussian randomness. Berry \cite{berry2002statistics} proposed that if a compact planar domain $\Omega \subset \mathbb{R}^2$ is chaotic (in terms of billiard dynamics), then the local behavior of a Dirichlet eigenfunction with eigenvalue $\lambda$ can be modeled by a Gaussian random field $F_{\lambda}$ satisfying $-\Delta F_{\lambda} = \lambda F_{\lambda}$ almost surely. He showed that the nodal set of $F_{\lambda}$ in $\Omega$ has average length $C \cdot \mathrm{area}(\Omega)\sqrt{\lambda}$ and variance of order $\mathrm{area}(\Omega)\log(\lambda)$ \cite{berry2002statistics}. Since then, nodal sets of random wave models and random eigenfunctions have been extensively studied, with notable results for Laplace eigenfunctions on the two-dimensional sphere \cite{nazarov2009number,nazarov2020fluctuations} and torus \cite{rudnick2008volume,Wigman2013nodal}. It was long believed that in all such cases, the properly normalized nodal length has a limiting Gaussian distribution. However, recent remarkable non-universality results identified sequences of eigenvalues for which the limiting nodal length distribution exists but is non-Gaussian \cite{marinucci2016non,Nourdin2019}. Our results reveal a comparable breakdown of universality.


We now turn to our main objects of interest, the nodal count of a matrix, or equivalently, of a finite graph. 
Given a real symmetric $n\times n$ matrix $A$, we associate a weighted signed graph $G$ with $n$ vertices and an edge $(i,j)$ for each pair $i<j$ with non-zero off-diagonal entry $A_{ij}\ne0 $. 
Every edge has an edge weight $|A_{ij}|$, and the edge sign is $-\mathrm{sgn}(A_{ij})$. Thus, for example, if $A$ is a graph Laplacian, then $G$ is unweighted and unsigned. Sort the eigenvalues of a $A$ in increasing order, $\lambda_{1}\le\lambda_{2}\le\ldots\le \lambda_{n}$, and denote the $k$-th eigenvector by $\bm \varphi^{(k)}$. 

\begin{definition} \label{def:counts}
The \emph{nodal count} $\phi(A,k)$ is the number of edges $(i,j)$ of $G$ for which $\varphi^{(k)}$ changes sign with respect to the sign of the edge, namely $A_{ij}\bm \varphi^{(k)}_i \bm \varphi^{(k)}_j>0$ . Under the generic assumption that $\bm \varphi^{(k)}_j\ne 0$ for all $k,j$, we write   
\[\phi(A,k):=\left|\{i<j,\ :\ A_{ij}\bm\varphi^{(k)}_{i}\bm\varphi^{(k)}_{j}>0\}\right|=\frac{1}{2}\sum_{i<j}(1+\mathrm{sgn}(A_{ij}\bm \varphi^{(k)}_i \bm \varphi^{(k)}_j)).\]

\end{definition}
Fiedler was the first to introduce and study this count \cite{Fiedler1975}, who showed that an analogue of Sturm's theorem holds: if $G$ is a tree, then  $\phi(A,k) = k-1$ for all $k$. Conversely, Band  \cite{band2014nodal} showed that if $\phi(A,k)=(k-1)$ for all $k$, then $G$ must be a tree. As in the continuous setting, there are no such uniform estimates for general graphs. Instead, Berkolaiko \cite{berkolaiko2008lower} generalized Fiedler's result and showed that for general graphs the following inequality always holds:
\[0\le\phi(A,k)-(k-1)\le\beta, \]
where $\beta=\beta(G)\ge 0$ is the first Betti number of the graph, and $\beta(G)=0$ if and only if $G$ is a tree. In the presence of randomness, if $G$ is unweighted, and $\bm\varphi$ is replaced by a random Gaussian vector, then the nodal count is binomial and converges to a Gaussian as $n\to\infty$ \cite{bandnodal}. Thes heuristics, together with numerical simulations and analogous results for quantum graphs \cite{gnutzmann2003nodal,alon2024universality}, led to the common belief that the nodal surplus sequence $\sigma(A,k):=\phi(A,k)-(k-1)$ for $k=1,\ldots,n$ should behave like a Gaussian centered at $\beta/2$ as $\beta\to\infty$. We introduce some notation for this matter:
\begin{notations}
    For $\bm{x} \in \mathbb{R}^n$, we define $\av(\bm{x}) = \frac{1}{n}\sum_{i =1}^n \bm{x}_i$, $\std(\bm{x}) = \left( \frac{1}{n}\sum_{i=1}^n (\bm{x}_i - \mathrm{avg}(\bm{x}))^2\right)^{1/2}$, the empirical distribution of its coordinates \[\emp(\bm{x}):=\frac{1}{n}\sum_{k=1}^n\delta_{\bm{x}_{k}},\]
and the normalized empirical measure 
\[\Nemp(\bm{x}):=\mathrm{emp}\left(\frac{\bm{x}-\mathrm{avg}(\bm{x})\bm{1}}{\mathrm{std}(\bm{x})}\right).\]
\end{notations}

Let $\phi(A)=(\phi(A,k))_{k=1}^n$ and $\sigma(A) = (\sigma(A,k))_{k=1}^n$. The common belief was that $\emp_{\mathcal{N}}(\sigma(A))$ is very close to a Gaussian centered at $\beta/2$ with variance of order $\beta$, for large $n$ and $\beta$. This was disproved by two of the authors \cite{alon2024average}, as they proved the following sharp bounds 
\[\tfrac{\beta}{n}\le\av(\sigma(A))\le \beta -\tfrac{\beta}{n}.\]
Examples of matrices saturating the upper and lower bounds were given for any possible choice of $n$ and $\beta$ \cite{alon2024average}. These bounds were based on Berkolaiko's nodal-magnetic theorem \cite{berkolaiko2013nodal,colin2013magnetic}, which states that the nodal surplus is a Morse index of the eigenvalue. More precisely, by assigning to each edge $(i,j)$ an angle $\theta_{ij}=-\theta_{ji}$, the eigenvalues of the Hermitian matrix $(A_{\theta})_{rs}=e^{i\theta_{rs}}A_{rs}$ are functions of $\theta$, that have critical points when $\theta_{ij}\in\{0,\pi\}$ for all edges. At such a critical point, the Morse index of $\lambda_{k}$ turns out to precisely equal the nodal surplus $\sigma(A_{\theta},k)$. Note that such a critical point corresponds to a signing $e^{-i\theta_{ij}} \in \{-1,1\}$ of the edges of the graphs. Moreover, since these signs can be taken to be random and independent from one another, common wisdom about Gaussian universality suggests a possible refinement to the above-mentioned Gaussian conjecture. 
Encouraged by numerical evidence, as well as theoretical results which we detail below, this conjecture was formalized in \cite{alon2023morse}. Specifically, the first author and Goresky conjectured that when a graph is randomly signed, namely $A_{\theta}$ with random $\theta_{ij}\in\{0,\pi\}$, the nodal surplus does converge to a Gaussian.
To be precise, given a fixed real symmetric matrix $A$, a random signing $(A_{\mathrm{sgn}})_{ij}=\pm A_{ij}$ is a symmetric matrix with $|(A_{\mathrm{sgn}})_{ij}| = |(A)_{ij}|$, and with random signs drawn independently and uniformly from $\{-1,1\}$. 
It was conjectured that \cite[p.3]{alon2023morse}:
\begin{center}
    \emph{For any graph $G$ and a generic matrix $A$ supported on $G$,
\begin{equation}\label{conj}
  \Nemp(\sigma(A_{\mathrm{sgn}}))\to N(0,1), \quad\text{as}\quad \beta(G)\to\infty.  
\end{equation}}
\end{center}
As mentioned, this conjecture was confirmed for some specific cases. In \cite{alon2023morse} it is shown that if $G$ is the complete graph and the diagonal entries of $A$ are distinct and sufficiently large, then $\sigma(A_{\mathrm{sgn}},k)\sim \mathrm{Bin}(\beta,\tfrac{1}{2})$ for all $k$, and in particular $\emp(\sigma(A_{\mathrm{sgn}}))\sim \mathrm{Bin}(\beta,\tfrac{1}{2})$. 

For sparser graphs, according to \cite{alon2024nodal}, if the simple cycles of the graph are pairwise disjoint, then similarly, $\sigma(A_{\mathrm{sgn}},k)\sim \mathrm{Bin}(\beta,\tfrac{1}{2})$ for all $k$, and so $\emp(\sigma(A_{\mathrm{sgn}}))\sim \mathrm{Bin}(\beta,\tfrac{1}{2})$, for any generic choice of $A$ supported on $G$. In all mentioned cases the Gaussianity conjecture \ref{conj} holds as $\beta \to \infty$, and follows from standard probabilistic considerations about the binomial distribution.

In this work, we derive a non-universality result, akin to the continuous results of \cite{marinucci2016non,Nourdin2019}, which refutes the Gaussianity conjecture \ref{conj}. Specifically, in Corollary \ref{cor:disprove} we show that if $A$ is a random $n\times n$ GOE matrix,
\[ \lim_{n\to\infty}\Nemp(\sigma(A_{\mathrm{sgn}}))= \rho_{\mathrm{s.c}}\ne N(0,1)\quad\text{almost surely},\]
where $\rho_{\mathrm{s.c}}$ is the normalized semicircle distribution.
To facilitate this result, we will investigate the nodal count of random real symmetric matrices of size $n\times n$ for large $n$. As mentioned, our goal is to study GOE random matrices. However, our results hold more generally and apply to any orthogonally invariant ensemble. An ensemble of random real symmetric matrices of size $n\times n$ is called an \emph{orthogonally invariant ensemble} if it is invariant under conjugation by orthogonal matrices. We defer the formal definitions to Section \ref{sec:ortho} below, but mention for now that a key point is that these ensembles are uniquely determined by the distribution of their eigenvalues. So, if $p(\bm{\Lambda})$ is a distribution over eigenvalues $\bm{\Lambda}=(\lambda_{1},\ldots,\lambda_{n})$, we will denote the corresponding orthogonally invariant ensemble by $\mathrm{OE}_n(p(\bm{\Lambda}))$. 

Our main result states that, for any orthogonally invariant ensemble satisfying some very mild regularity assumptions, the empirical nodal count distribution converges to a properly normalized version of its eigenvalue distribution.
    Below, we use $f(n)=\tilde{O}(n^\alpha)$ to denote $$|f(n)|<C_1 n^{\alpha}\log^{C_{2}}(n)\quad\text{for all $n>N_0$}\quad\text{for some positive constants $N_{0,}C_{1},\ C_{2}$}.$$

\begin{theorem}\label{thm:main}
Let $A \sim \mathrm{OE}_n(p(\bm{\Lambda}))$ satisfy Spectral Growth Bound \ref{spec_growth}. Then
\begin{equation}\label{eqn:expectation}
    \E \left[\phi(A,k) \right] = {n \choose 2} \left( \frac{1}{2} + \frac{\sqrt{2}}{\pi^{3/2}} \E \left[ \frac{\lambda_k - \mathrm{avg}(\bm{\Lambda})}{\mathrm{std}(\bm{\Lambda})} \right]n^{-1/2}  + \tilde O\left(n^{-3/2}\right) \right)
\end{equation}
and 
\begin{equation}\label{eqn:variance}
    \mathrm{Var}\left[\phi(A,k) \right] = \tilde O(n^{5/2})
\end{equation} uniformly over $k \in [n]$.
\end{theorem}
At a technical level, the orthogonal invariance allows us to represent $A_{ij}\bm \varphi_i^{(k)}\bm \varphi_j^{(k)}$ as a low-degree polynomial over the orthogonal group, equipped with its Haar measure. However, the existence of the highly non-linear $\mathrm{sgn}$ function in the definition of $\phi(A,k)$ precludes the use of classical tools, such as the Weingarten calculus, for the calculation of $\E\left[\phi(A,k) \right]$ and $\mathrm{Var}\left[\phi(A,k) \right]$. Instead, our approach, which could be of potential independent interest, goes by a reduction to integration over Gaussian variables. While such a reduction would typically require considering a Gaussian space with dimension depending on $n$, we show that through careful conditioning, it is enough to consider a bounded number of Gaussians. In practice, this reduction allows us to estimate correlations between quadratic forms on the orthogonal group with very high precision, see Proposition \ref{prop:main_prop}. As we shall explain, $\phi(A,k)$ can be represented in this form which leads to Theorem \ref{thm:main}.

To better understand Theorem \ref{thm:main} consider the following normalization
\[\phi_{\mathcal{N}}(A,k):=\frac{\pi^{3/2}}{\sqrt{2}}\left(\frac{\phi(A,k)}{{n \choose 2}}-\frac{1}{2}\right) n^{1/2}.\]
Note that $\phi(A,k)$ is invariant to scaling and translating the matrix $A\mapsto c_{1}A-c_{2}\mathrm{I}$ for $c_1>0$ and $c_{2}\in\R$. This invariance is also reflected in the quantitative estimate appearing in Theorem \ref{thm:main}, where we consider the normalized eigenvalue. Therefore, we can and always will assume that $\mathrm{avg}(\bm{\Lambda}) = 0$ and $\mathrm{std}(\bm{\Lambda}) = 1$. Under this normalization, Theorem \ref{thm:main} can be rephrased as
\begin{equation*}
    \E \left[\phi_{\mathcal{N}}(A,k)\right] =  \E[\lambda_k]  + \tilde O\left(n^{-1}\right),\quad\text{and}\quad \mathrm{Var}\left[\phi_{\mathcal{N}}(A,k)\right]=\tilde O\left(n^{-1/2}\right),
\end{equation*}
which immediately affords the following corollary concerning the asymptotics of the limiting empirical distribution $\emp_{\mathcal{N}}\left(\phi(A)\right)$, and as a result, also $\emp_{\mathcal{N}}\left(\sigma(A)\right)$. Recall that $\sigma(A)=\phi(A)+O(n)$ by definition, and that a sequence $\mu_{n}$ of random probability measures is said to \emph{converge almost surely} to a deterministic probability measure $\rho$ if, for any compactly supported continuous function $\varphi$, the random variables $\int\varphi \, d\mu_{n}$ converge to the number $\int\varphi \, d\rho$ almost surely. We denote such convergence by $\mu_{n}\underset{\mathrm{a.s.}}{\longrightarrow}\rho$.

\begin{corollary}[General convergence] \label{cor:conver}
Under the assumptions of Theorem \ref{thm:main}, suppose further that $\mathrm{avg}(\bm{\Lambda}) = 0$ and $\mathrm{std}(\bm{\Lambda}) = 1$, and that there exists some probability measure $\rho$ on $\R$ such that
$$\mathrm{emp}(\bm{\Lambda}) \underset{\mathrm{a.s.}}{\longrightarrow} \rho, \quad \text{as $n\to\infty$}.$$
Then, the normalized nodal count $\phi_{\mathcal{N}}(A)=\left(\frac{\pi^{3/2}}{\sqrt{2}}\left(\frac{\phi(A,k)}{{n \choose 2}}-\frac{1}{2}\right)n^{1/2}\right)_{k\in[n]}$ converges as well:
$$\emp\left(\phi_{\mathcal{N}}(A)\right)\underset{\mathrm{a.s.}}{\longrightarrow} \rho.$$
The normalized nodal surplus $\sigma_{\mathcal{N}}(A)=\left(\frac{\pi^{3/2}}{\sqrt{2}}\left(\frac{\sigma(A,k)}{{n \choose 2}}-\frac{1}{2}\right)n^{1/2}\right)_{k\in[n]}$ satisfies the same convergence
$$\emp\left(\sigma_{\mathcal{N}}(A)\right)\underset{\mathrm{a.s.}}{\longrightarrow} \rho.$$


\end{corollary} 
Corollary \ref{cor:conver} suggests a general machinery of constructing different matrices with different limiting distributions for $\sigma_{\mathcal{N}}(A)$. Specifically, we can choose an appropriate orthogonal ensemble, sample a matrix $A$ from this ensemble, and compute its normalized nodal surplus $\sigma_{\mathcal{N}}(A)$. 
We demonstrate this procedure in Figure \ref{fig:non_unimodal}, where we have chosen an orthogonal ensemble with a tri-modal eigenvalue distribution. The figure contains histograms showing $\emp\left(\phi_{\mathcal{N}}(A)\right)$ and $\emp(\bm{\Lambda})$, in this case. These histograms stand in sharp contrast to the usual unimodal histograms ubiquitous in the nodal count literature. 

\begin{figure}[t]
\centering
  \begin{subfigure}[t]{0.48\textwidth}
    \centering
    \includegraphics[width =\textwidth]{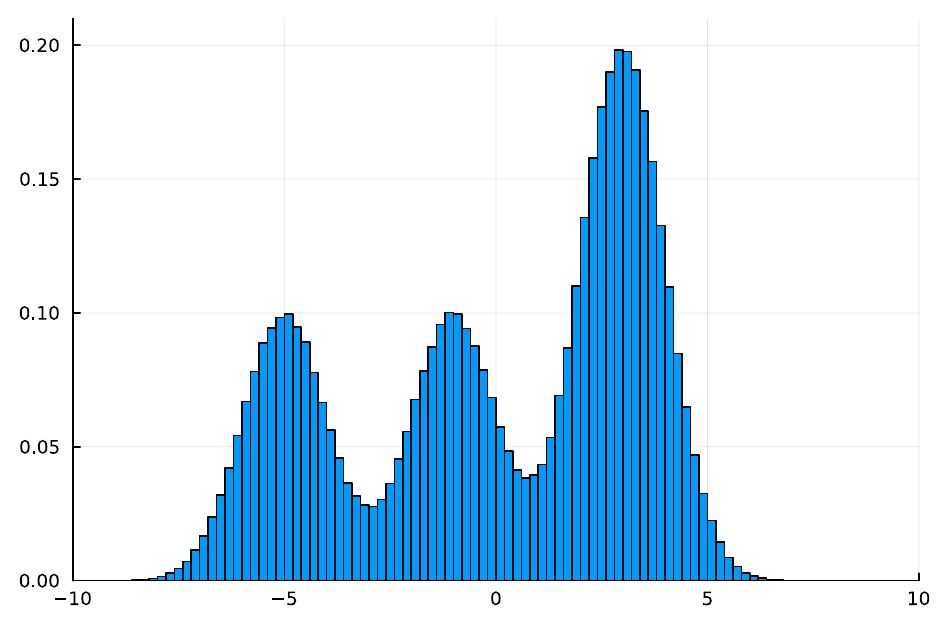}
    \caption{Histogram of the spectrum}
    \label{fig:P_Lambda}
  \end{subfigure}\hfill
  \begin{subfigure}[t]{0.48\textwidth}
    \centering
\includegraphics[width =\textwidth]{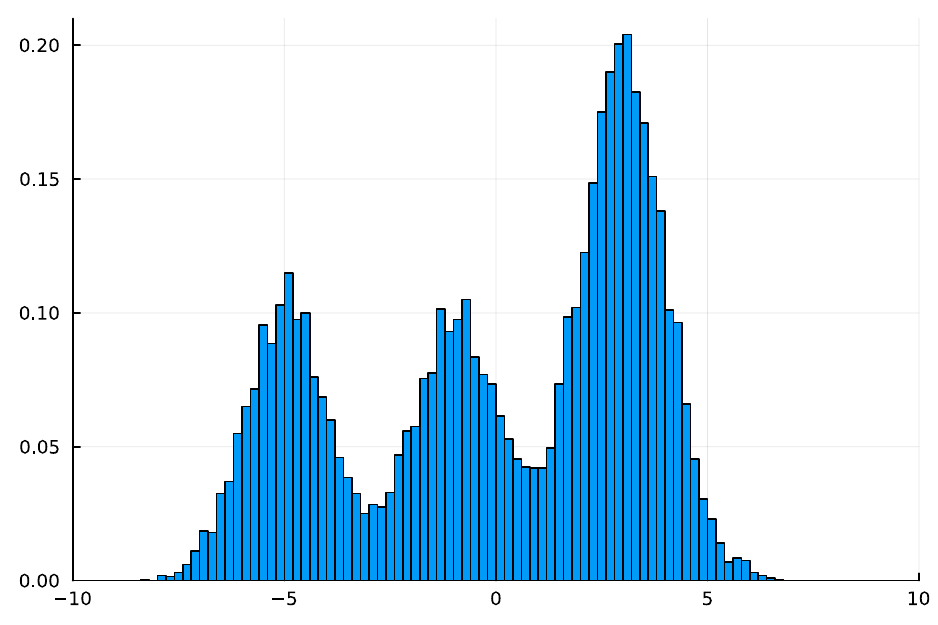}
    \caption{Histogram of the normalized nodal count.}
    \label{fig:sigma_hist}
  \end{subfigure}
 \caption{\textbf{Eigenvalue–nodal count agreement for a Gaussian mixture.} 
Numerical experiment for a \emph{single} $n\times n$ random matrix 
$A = O^{\mathsf T}\!\operatorname{diag}(\Lambda)O$ with $n=10^4$, 
where $O$ is Haar orthogonal and the entries of $\Lambda$ are sampled independently from the Gaussian mixture distribution 
$f(x) = \tfrac{1}{4\sqrt{2\pi}}\!\left(e^{-\tfrac{(x+5)^2}{2}} + e^{-\tfrac{(x+1)^2}{2}} + 2 e^{-\tfrac{(x-3)^2}{2}}\right)$. 
We emphasize that this figure shows results for a single sampled matrix, not an average over many realizations.}
\label{fig:non_unimodal}
\end{figure}

\quash{
\begin{figure}[ht]
  \centering
  \begin{subfigure}[t]{0.48\textwidth}
    \centering
    \includegraphics[width=\textwidth]{fig1_PLambda_hist_n100.pdf}
    \caption{$P(\bm{\Lambda})$ for $n=100$ (histogram with asymmetric three bumps).}
    \label{fig:P_Lambda}
  \end{subfigure}\hfill
  \begin{subfigure}[t]{0.48\textwidth}
    \centering
    \includegraphics[width=\textwidth]{fig2_sigma_hist_n100.pdf}
    \caption{Histogram of $\phi(A,k) - (k-1)$ for $n=100$,
      with a vertical marker at $\beta/2$ on the $x$-axis.}
    \label{fig:sigma_hist}
  \end{subfigure}
  \caption{Comparison of histograms: both panels are qualitatively similar.}
  \label{fig:non_unimodal}
\end{figure}}
\quash{Given a real symmetric matrix $A$, we say that $A_{\mathrm{sgn}}$ is a random signing of $A$ when $(A_{\mathrm{sgn}})_{ij}=(A_{\mathrm{sgn}})_{ji}=\eps_{ij}A_{ij}$ with $\eps_{ij}$ i.i.d Rademacher random variables for all $i< j$.}  

As mentioned, Corollary \ref{cor:conver} is of particular interest when $A$ is a GOE matrix. In that case, the symmetries of the Gaussian distribution imply that $A_{\mathrm{sgm}}$ has the same distribution as $A$, which leads to a refutation of \eqref{conj}.
\begin{corollary}[GOE convergence]\label{cor:disprove}
Let $A$ be a $\mathrm{GOE}_{n}$ matrix and let $A_{\mathrm{sgn}}$ be a random signing of $A$. Then $A_{\mathrm{sgn}}$ is also a $\mathrm{GOE}_{n}$ matrix, 
so 
$$\Nemp\left(\sigma(A_{\mathrm{sgn}})\right)\underset{\mathrm{a.s.}}{\longrightarrow} \rho_{\mathrm{s.c}}\ne N(0,1),$$
where $\rho_{\mathrm{sc}}$ is the normalized semicircle distribution. In particular, Conjecture \cite[pp. 1227]{alon2023morse} is false.    
\end{corollary}
See Figure \ref{fig3:a}, where the convergence $\Nemp\left(\sigma(A_{\mathrm{sgn}})\right)\to \rho_{\mathrm{s.c}}$ is shown in a numerical experiment.

\begin{figure}[t]
  \centering
  \begin{subfigure}[t]{0.48\textwidth}
    \centering
    \includegraphics[width=\linewidth]{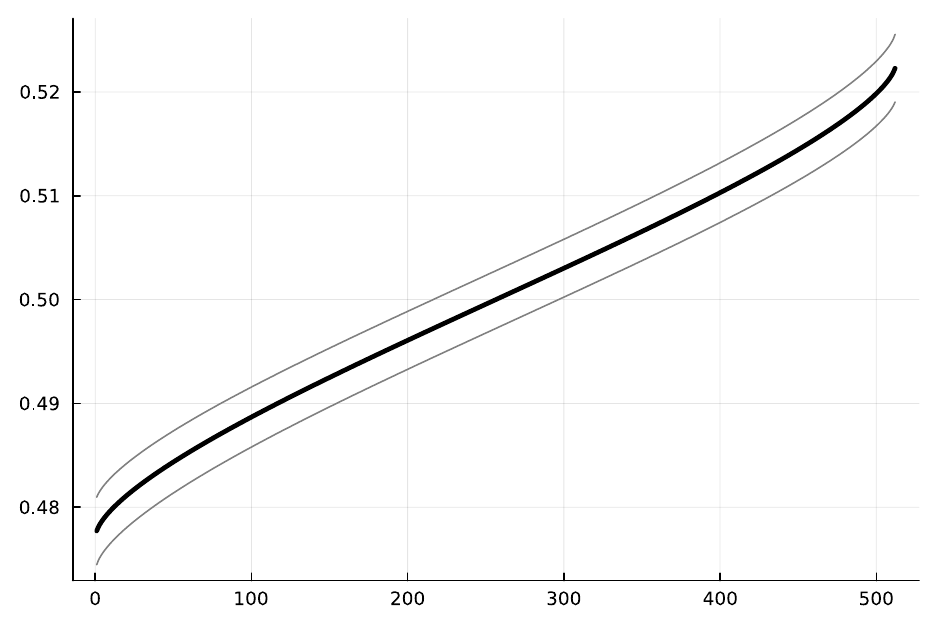} \caption{{}\hspace{-5 mm}\phantom{}} \label{fig2:a}
  \end{subfigure}\hfill
  \begin{subfigure}[t]{0.48\textwidth}
    \centering
    \includegraphics[width=\linewidth]{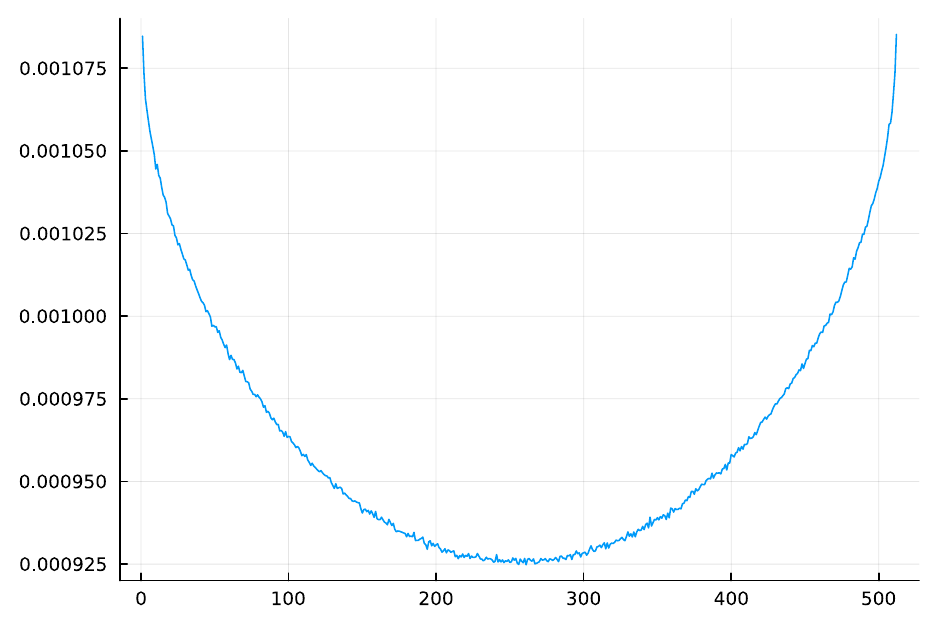} \caption{{}\hspace{-5 mm}\phantom{}} \label{fig2:b}
  \end{subfigure}
  \caption{\textbf{Concentration of the nodal count for GOE matrices.}
 In this experiment we sampled $10^6$ independent random GOE matrices of size $n=2^9$.
On the left, the thick curve shows the empirical mean $\mathbb{E}[\phi(A,k)/\binom{n}{2}]$, while the thin curves indicate $\pm3$ standard deviations.
The mean takes values between $0.48$ and $0.52$, consistent with fluctuations of order $n^{-1/2}$ around $1/2$.
On the right, the standard deviation $\std[\phi(A,k)/\binom{n}{2}]$ is plotted as a function of $k$, taking values between $0.925\times10^{-3}$ and $1.075\times10^{-3}$, of order $n^{-1}$, which means $\Var[\phi(A,k)]$ is of order $n^2$.
}
  \label{fig:combined_var_mean}
\end{figure}

\subsection*{\bf CLT conjecture and further questions:} Corollary \ref{cor:disprove} establishes that randomly signing the edges of the matrix is not enough to ensure universality of the nodal surplus count, and different randomly signed matrices will exhibit different behaviors.

To probe finer universality properties, we focus on the normalized fluctuations of the $k$-th nodal count. That is, let $A$ be random $n\times n$ GOE matrix, and fix $k$. Consider the normalized nodal count
$$\overline{\phi(A,k)}:=\frac{\phi(A,k) - \E\left[\phi(A,k)\right]}{\sqrt{\mathrm{Var}(\phi(A,k))}}.$$ 
It is reasonable to expect that $\overline{\phi(A,k)}$ will be approximately Gaussian. As motivation, Theorem 3.2 (7) in \cite{alon2023morse} provides an explicit sufficient condition on a matrix $A$ that ensures $\overline{\phi(A_{\mathrm{sgn}},k)}$ is approximately Gaussian:  
Given $A$ real symmetric $n\times n$ matrix, for any $\theta$ real antisymmetric matrix define the Hermitian matrix $(A_{\theta})_{rs}=A_{rs}e^{i\theta_{rs}}$. If $A$ is such that for every $\theta$, all eigenvalues of $A_{\theta}$ are simple with nowhere-vanishing eigenvectors, then $A$ has $\sigma(A_{\mathrm{sgn}},k)\sim \mathrm{Bin}(\beta,\frac{1}{2})$ for all $k$. 

In particular, a standard argument on convergence of binomial to Gaussian allows us to conclude that any matrix $A$ satisfying the assumption of Theorem 3.2 (7) in \cite{alon2023morse} has
\[d_\mathrm{Kol}\left(\overline{\phi(A_{\mathrm{sgn}},k)},N(0,1)\right)=o(1),\]
as $n$ grows, independently of $k$ and the choice of $A$, and with $d_\mathrm{Kol}$ standing for Kolmogorov distance.

However, the $k$ independence of $\sigma(A_{\mathrm{sgn}},k)\sim \mathrm{Bin}(\beta,\frac{1}{2})$ implies $\sigma(A_{\mathrm{sgn}})\sim \mathrm{Bin}(\beta,\frac{1}{2})$, which is not the case for a random GOE matrix, by 
Corollary \ref{cor:disprove}. That is, a random GOE matrix $A$ \textbf{does not} satisfy the assumption of Theorem 3.2 (7) in \cite{alon2023morse}.
 Nevertheless, recall that $\phi(A,k) = \frac{1}{2}\sum_{i<j}\big(1+\mathrm{sgn}(A_{ij}\bm \varphi^{(k)}_i \bm \varphi^{(k)}_j)\big)$ is a sum of simple random variables. While these random variables are not independent, their joint law is invariant to certain permutations, and they form an exchangeable double array. That is, the law remains unchanged when either permuting the rows or the columns. We verify, for a random GOE matrix $A$, that $\mathrm{dist}\left(\overline{\phi(A,k)},N(0,1)\right)\to 0$ numerically, as seen in Figure \ref{fig3:b}.

Given the numerical results and the above intuition, we raise the following modest conjecture.
\begin{conjecture}[CLT] \label{conj:CLT}
Let $A$ be a random $n\times n$ matrix drawn from the Gaussian Orthogonal Ensemble. Then, for any $k\in[n]$,  
\[\mathrm{dist}\left(\overline{\phi(A,k)},N(0,1)\right)=o(1),\] 
as $n$ grows, uniformly in $k$, and using any reasonable distance between random variables.
\end{conjecture}

A bolder conjecture would involve a much larger class of signed matrices. In line with the original conjecture in \eqref{conj}, one could consider any sequences of matrices with diverging Betti numbers. The paper \cite{alon2025smooth} raises such a question and contains further discussion.

Our proof falls short of addressing Conjecture \ref{conj:CLT}. One reason is that our estimate $\mathrm{Var}\left[\phi(A,k) \right] = \tilde O(n^{5/2})$ from Theorem \ref{thm:main} is not sharp. Numerical simulations suggest that the true growth rate is $\mathrm{Var}\left[\phi(A,k) \right] = \tilde O(n^{2})$  as seen in Figure \ref{fig3:c}. 

A natural approach towards the conjecture would be to understand the moments, at least the lower order ones, of $\overline{\phi(A,k)}$, which would require tightening our analysis. This is further emphasized in Figure \ref{fig2:b}, which shows that there are lower-order dependencies of $\mathrm{Var}[\phi(A,k)]$ with respect to $k$. Our coarse bounds are unable to capture this quadratic behavior. In light of this, we view a better understanding of this plot as a first step towards answering Conjecture \ref{conj:CLT}. While, in principle, our approach is well-suited for making these finer estimates, it would make the proof much more technical. Thus, for the sake of readability and to reduce the technicality of this paper, we chose not to pursue this direction in the present work, and leave it as an interesting question for the future.

\begin{figure}[t]
  \centering
  \begin{subfigure}[t]{0.32\textwidth}
    \centering
    \includegraphics[width=\linewidth]{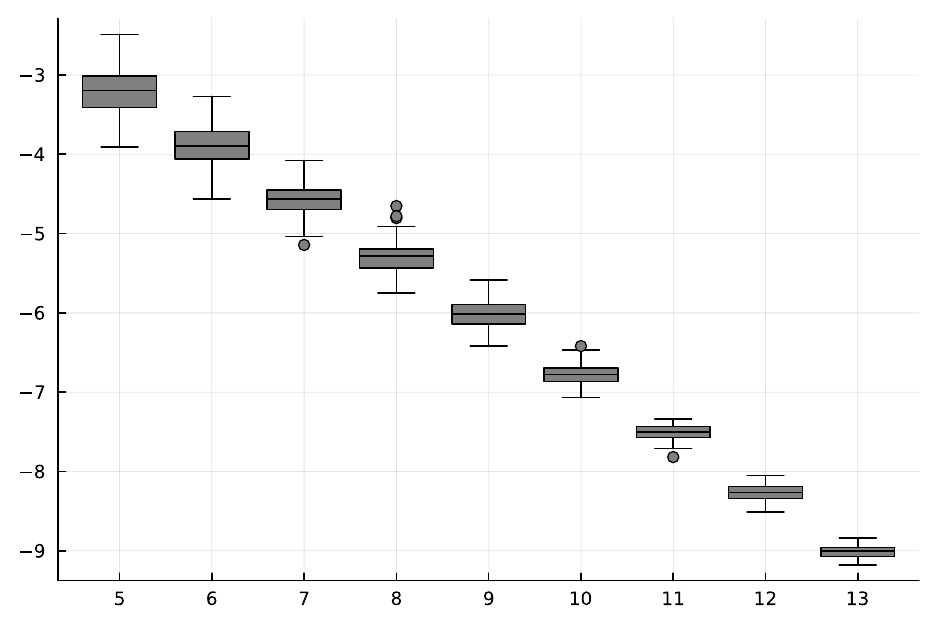}\caption{} \label{fig3:a}
  \end{subfigure}\hfill
  \begin{subfigure}[t]{0.32\textwidth}
    \centering
    \includegraphics[width=\linewidth]{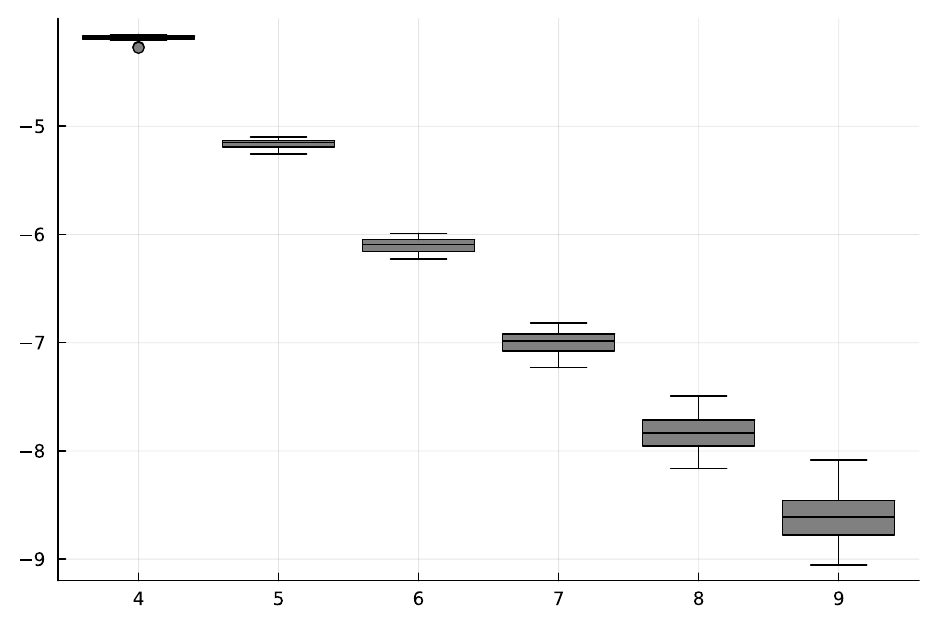} \caption{} \label{fig3:b}
  \end{subfigure}\hfill
  \begin{subfigure}[t]{0.32\textwidth}
    \centering
    \includegraphics[width=\linewidth]{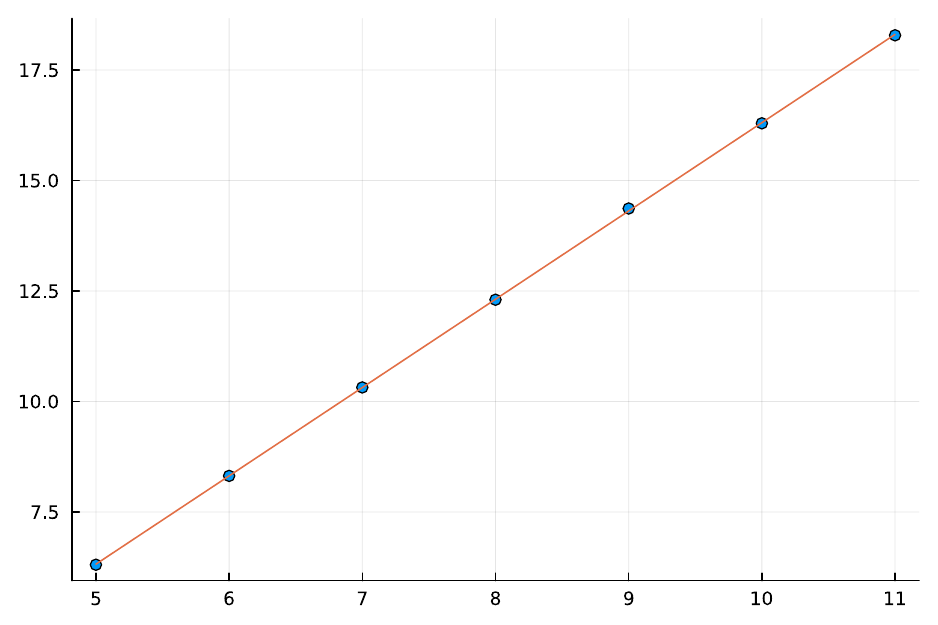}\caption{} \label{fig3:c}
  \end{subfigure}

  \caption{\textbf{Asymptotic statistics of the nodal count for GOE matrices.}
  \small
  (\textbf{A})~Log--log boxplot of the 1-Wasserstein distance between 
  $\emp(\phi_{\mathcal{N}}(A))$ and the semicircle distribution, 
  as a function of $\log n$ for $n=2^5,2^6,\dots,2^{13}$.
  Each box summarizes $100$ independent samples of $\mathrm{GOE}_n$ matrices: 
  the central line indicates the median, the box the interquartile range, 
  whiskers extend to $1.5$ times the interquartile range, and circles mark outliers. 
  (\textbf{B})~Similar log--log boxplot of the Kolmogorov–Smirnov (KS) distance between $\phi(A,k)$ and a Gaussian distribution of the same mean and variance, for $n=2^4,\dots,2^9$. Each box summarizes the distribution of KS distances over all $k\in\{1,\ldots,n\}$, computed from $10^6$ independent GOE matrices.
  (\textbf{C})~Log--log plot of $\max_k\mathrm{Var}[\phi(A,k)]$ versus $\log n$, 
  estimated from $10^4$ independent samples; the fitted slope $1.998$ and intercept $-3.670$ 
  suggest $\mathrm{Var}[\phi(A,k)]=\tilde{O}(n^2)$.}
  \label{fig:asymptotic_stats}
\end{figure}


\section{Preliminaries and Proofs of Corollaries}
\subsection{Some asymptotic notation}
Throughout, we shall use the standard big $O$ and $\Omega$ for asymptotic notation. Since we care about polynomial bounds, we shall use $\tilde{O}$ to hide poly-logarithmic factors. Formally, $\tilde{O}(f(n))$ means $O(\log(n)^cf(n))$ for some $c>0$.
For vectors, we shall deviate slightly from the standard usage of the notation. Thus if $\nu$ is some vector, we shall write $\nu = \tilde{O}(f(n))$, when $\frac{\|\nu\|}{\sqrt{n}} =  \tilde{O}(f(n))$, The main point is that with this definition, for \emph{most} unit vectors $\theta$,
\begin{equation} \label{eq:asynot}
    \langle \nu,\theta\rangle = \tilde{O}(f(n)).
\end{equation}
To avoid possible confusion, we will recall this notation when using it.
Similarly, if $N$ is any random variable, we write $N = \tilde{O}(f(n))$ when $\sqrt{\E[N^2]} = \tilde{O}(f(n))$. In particular by the Cauchy-Schwartz inequality, if $N = \tilde{O}(f(n))$ and $N' = \tilde{O}(g(n))$, then $\E[N\cdot N']  = \tilde{O}(f(n)g(n))$.

\subsection{Preliminaries on real orthogonal ensembles} \label{sec:ortho}
Here, we shall give a brief overview of orthogonally invariant ensembles. In this work, our focus is on extracting meaningful statistics from an orthogonal ensemble when the eigenvalues of the ensemble are known. This is different from the traditional setting, where one wishes to understand the eigenvalues from the law of the random matrices. Therefore, we shall introduce the orthogonally invariant ensembles in an ad hoc and somewhat nontraditional way. The reader is referred to \cite{deift2009orthogonal} for a more comprehensive treatment, as well as references to all claims made in the following.

Let $p$ be a probability distribution on $\R^n$. For $\bm{\Lambda}\sim p$ we order its entries by $\lambda_1\leq\lambda_2,\dots\leq \lambda_n$. For $\Phi \sim \mathrm{Haar}(\mathcal{O}(n))$, a random orthogonal matrix, independent from $\bm{\Lambda}$, we say that the random matrix $M = \Phi\mathrm{diag}(\bm{\Lambda})\Phi^T$ is drawn from the (real) orthogonal ensemble with eigenvalue distribution $p$, and denote it by $\mathrm{OE}_n(p(\bm{\Lambda}))$. Equivalently, a symmetric $M$ is drawn from some orthogonal ensemble if $\Psi M \Psi^*$ has the same distribution as $M$ for any orthogonal matrix $\Psi$. 

Perhaps the best-known orthogonal ensemble is the Gaussian Orthogonal Ensemble (GOE), in which up to symmetries $M$ has independent Gaussian entries. Specifically, $M_{i,j} \sim N(0,\frac{1}{n})$ when $i\neq j \in [n]$ and $M_{i,i} \sim N(0, \frac{2}{n})$ for $i \in [n]$. In this case, the rotational invariance of the Gaussian immediately implies that $M$ is drawn from some orthogonal ensemble. In fact, we can show that the eigenvalue distribution is explicitly given by
\begin{equation} \label{eq:GOEeigs}
    p_{\mathrm{GOE}(n)}(\bm{\Lambda}) \propto \exp\left(-\frac{1}{4}\sum\limits_{i=1}^n\lambda_i^2\right) \prod_{i<j} |\lambda_i-\lambda_j|.
\end{equation}
A cornerstone result in random matrix theory states that the empirical eigenvalue distribution of the GOE stabilizes as $n\to \infty$. That is, if $\bm{\Lambda}_n\sim  p_{\mathrm{GOE}(n)}$,
then 
\begin{equation} \label{eq:wigner}
    \emp(\bm{\Lambda}_n) \underset{\mathrm{a.s.}}{\longrightarrow} \rho_{\mathrm{s.c}},
\end{equation}
where $\rho_{\mathrm{s.c}}$ is the semi-circle law, with explicit density
$\frac{d\rho_{\mathrm{s.c}}}{dx}(x) = \frac{1}{2\pi}\sqrt{4-x^2}\on{|x|\leq 2}.$ 

Of course, there are many other examples of orthogonal ensembles. For example, if $M$ has independent Gaussian entries, then $MM^*$ is orthogonally invariant, and its empirical spectral distribution converges to the Marchenko-Pastur law. In general, the definition above allows the construction of many different examples, as we've demonstrated in Figure \ref{fig:non_unimodal}.

For our result, we will need to impose a mild regularity condition on the allowed eigenvalue distributions, which we now define.
\begin{definition}[Spectral growth bound]\label{spec_growth}
A random orthogonally invariant ensemble $\mathrm{OE}_n(p(\bm{\Lambda}))$ is said to satisfy the spectral growth bound if 
\[\left\| \frac{\bm{\Lambda} - \mathrm{avg}(\bm{\Lambda})}{\mathrm{std}(\bm{\Lambda})} \right\|_{\infty} = \tilde O(1)
\]
with probability at least $1 - \tilde O\left({n^{-3/2}}\right)$.
\end{definition}
To expand a bit on Definition \ref{spec_growth} and its uses, recall that if $M \sim \mathrm{OE}_n(p(\bm{\Lambda}))$, then for any $k \geq 0$, $\mathrm{Tr}(M^k) = \sum\limits_{i=1}^n\lambda_i^k$. In that case, the condition dictates that once we normalize $\mathrm{Tr}(M) = 0$ and $\mathrm{Tr}(M^2) = n$, $M$ must satisfy 
$\|M\|_{\mathrm{op}} \leq C\log(n)^c$ and $\mathrm{Tr}(M^k) \leq C\log(n)^{(k-2)c}n$, for some constants $C,c>0$, with non-negligible probability. Below, we will use this consequence with $k=4$.

For the GOE, it is straightforward to verify, using \eqref{eq:GOEeigs}, that the condition is satisfied. A tighter analysis actually shows that we can choose $c = 1$, see \cite[Chapter 3.1]{anderson2010introduction} for example. Similar estimates hold whenever the tails of the eigenvalue distribution are not too heavy.
\subsection{Preliminaries on Gaussian concentration}
Our proof of Theorem \ref{thm:main} will require us to estimate some integrals on the orthogonal group. As we shall demonstrate in the proof, it will be beneficial to reduce this calculation to Gaussian space. Towards that, we collect some useful results in this section.

The first one is a basic concentration inequality for low-degree polynomials in Gaussian variables.
\begin{proposition}[{\cite[Theorem 6.7]{janson1997gaussian}}]\label{prop:gaussian_concentration}
Let $p$ be a degree $k$ polynomial and let $G$ be a standard Gaussian vector. There exists a universal constant $c_k$ such that
\[ \mathbb{P} \left( |p(X)| > t \; \E \big[ |p(X)|^2\big]^{1/2} \right) \le \exp\{-c_k t^{2/k} \} \qquad \text{for all } t \ge 2.\] 
\end{proposition}
From Proposition \ref{prop:gaussian_concentration} , we can deduce a general bound for quadratic forms in Gaussians. For that, if $M$ is a matrix, we introduce the notation $\|M\|_{\max}:= \max_{i,j}|M_{i,j}|$.
\begin{lemma}\label{lm:typical}
Let $\bm{x} \in \mathbb{R}^{n}$ be a fixed vector with $\|\bm{x}\|_2 \le n^{1/2} \log^{c_1} n$, $\|\bm{x}\|_4 \le n^{1/4} \log^{c_1} n$ and $G \in \mathbb{R}^{n \times p}$ be a standard Gaussian matrix, for some fixed constants $c_1$ and $p$. Then there exists a constant $c_2$ such that $\|G\|_{\max} \le \log^{c_2} n$ and
\[\left\| G^T \mathrm{diag}(\bm{x})^k G - \E\big[ G^T \mathrm{diag}(\bm{x})^k G \big] \right\|_{\max} \le \sqrt{n} \log^{c_2} n\]
for $k = 0,1,2$ with probability $1 - n^{-\log n}$.
\end{lemma}
\begin{proof}
By repeated application of Proposition \ref{prop:gaussian_concentration}, it suffices to show that for an arbitrary entry of $G^T \mathrm{diag}(\bm{x})^k G$, its variance is at most $n \log^{c_3} n$ for some constant $c_3$. This indeed follows immediately from the identity $\mathrm{Var}[(G^T D G)_{ij}]= (1+ \delta_{ij})\|D\|_F^2$. 
\end{proof}
As a typical example of Lemma \ref{lm:typical}, suppose that $\bm{\Lambda}$ satisfies the spectral growth condition, as in Definition \ref{spec_growth}, and is normalized as in Corollary \ref{cor:conver}. Then, if $\bg$ and $\hg$ are two independent standard Gaussians, we may consider $G$ in Lemma \ref{lm:typical} to have $\bg$ and $\hg$ as its two columns. In that case the Lemma implies the following bounds with high probability,
\begin{equation} \label{eq:specconc}
    \langle \bg,\bg\rangle - n|, \quad|\langle \bg,\hg\rangle|,\quad |\langle \bg, \mathrm{diag}(\bm{\Lambda})\bg\rangle|,\quad |\langle \bg, \mathrm{diag}(\bm{\Lambda})\hg\rangle| \leq \sqrt{n}\log^{c_2}n,
\end{equation}|
where we have used that the spectral growth condition implies $\|\bm{\Lambda}\|_4 \leq n^{1/4}\log^{c_1}n$ for some $c_1\geq 0$.
Of course, the Lemma applies equally well to any finite sequence of independent standard Gaussians.

\subsection{Proofs of corollaries}
Having defined the necessary notions, we can now explain how to derive the Corollary \ref{cor:conver}and Corollary \ref{cor:disprove} from Theorem \ref{thm:main}.
\begin{proof}[Proof of Corollary \ref{cor:conver}]
Suppose that $\bm{\Lambda}_n = \bm{\Lambda}$ satisfies
\begin{enumerate}
    \item  the assumption of Theorem \ref{thm:main},
    \item $\mathrm{avg}(\bm{\Lambda}_{n}) = 0$, $\mathrm{std}(\bm{\Lambda}_n) = 1$ with probability $1$, and
    \item there exist some probability measure $\rho$ on $\R$ such that
$$\mathrm{emp}(\bm{\Lambda}_n) \underset{\mathrm{a.s.}}{\longrightarrow} \rho, \quad \text{as $n\to\infty$}.$$
\end{enumerate}
Let $x_n=\E[\bm{\Lambda}_{n}]$ and let $\varphi$ be a bounded Lipschitz function on $\R$, with bound $M$ and Lipschitz constant $C$. The above convergence implies $\frac{1}{n}\sum_{k=1}^{n}\varphi(x_{n,k})\to\int\varphi d\rho$. For $A\in\mathrm{OE}_n(p(\bm{\Lambda}_n))$, the normalization and Theorem \ref{thm:main} gives \begin{equation*}
    \E \left[\phi_{\mathcal{N}}(A,k)\right] =  x_{n,k}  + \tilde O\left(n^{-1}\right),\quad\text{and}\quad \mathrm{Var}\left[\phi_{\mathcal{N}}(A,k)\right]=\tilde O\left(n^{-1/2}\right),
\end{equation*}
uniformly on $k$.
Fix $0<\eps<\frac{1}{4}$, Chebyshev's inequality (and the triangle inequality) now implies  $$\P[|\phi_{\mathcal{N}}(A,k)-x_{n,k}|>n^{-\eps}]=\tilde{O}(n^{2\eps-\frac{1}{2}})\le \delta_{n},$$
for some sequence $\delta_{n}\to 0$ (independent of $k$). For any $t>0$, Markov's inequality (and the triangle inequality) gives
\begin{align*}
    \P\left[\left|\frac{1}{n}\sum_{k=1}^n\varphi(\phi_{\mathcal{N}}(A,k))-\varphi(x_{n,k})\right|>t\right]\le & \frac{1}{t}\frac{1}{n}\sum_{k=1}^n\E\bigg[\big|\varphi(\phi_{\mathcal{N}}(A,k))-\varphi(x_{n,k})\big|\bigg]\\
    \le & \frac{1}{t}(Cn^{-\eps}+2M\delta_{n})\to 0.    
\end{align*}
In particular, the estimates continues to hold for $t = \frac{1}{\sqrt{n}^\eps}$. So, $\frac{1}{n}\sum_{k=1}^{n}\varphi(x_{n,k})\to\int\varphi d\rho$, and we conclude that
\[\P\left[\frac{1}{n}\sum_{k=1}^n\varphi(\phi_{\mathcal{N}}(A,k))\to \int\varphi d\rho\right]=1.\]
Using the Portmanteau Theorem (to go from any bounded Lipschitz to any bounded continuous) we conclude that $\mathrm{emp}(\phi_{\mathcal{N}}(A)) \underset{\mathrm{a.s.}}{\longrightarrow} \rho$, as needed. 

Since $\phi_{\mathcal{N}}(A,k))-\sigma_{\mathcal{N}}(A,k))=c(k-1)\frac{\sqrt{n}}{{n \choose 2}}$ is deterministic and is uniformly bounded by $\frac{c'}{\sqrt{n}}$ , the same argument gives $\mathrm{emp}(\sigma_{\mathcal{N}}(A)) \underset{\mathrm{a.s.}}{\longrightarrow} \rho$
\end{proof}
Corollary \ref{cor:disprove} follows immediately from Corollary \ref{cor:conver}, Wigner's Semi-Circle Theorem \eqref{eq:wigner}, and the following two facts:
\begin{enumerate}
    \item if $A$ is GOE then $A_{\mathrm{sgn}}$ is GOE,
    \item if $\mathrm{emp}(\sigma_{\mathcal{N}}(A)) \underset{\mathrm{a.s.}}{\longrightarrow} \rho$ and $\rho$ is symmetric around $0$, then $\mathrm{emp}_{\mathcal{N}}(\sigma(A)) \underset{\mathrm{a.s.}}{\longrightarrow} \rho$. 
\end{enumerate}\qed
\section{Signs of Quadratic Forms on the Orthogonal Group}
Our main tool for the proof of Theorem \ref{thm:main} is a precise estimate on the possible correlation between the signs of certain quadratic forms on the orthogonal group. As mentioned, the main difficulty in handling such integrals lies in the inclusion of the sign function. To address this difficulty, we reduce the computation on the orthogonal group to a computation in Gaussian space. For Gaussian variables, there is an explicit formula for the correlation of signs, sometimes called Grothendieck's identity or Sheppard's formula.
\begin{lemma}\label{lem: signArcsine}
Let $(X,Y)\in \R^2$ be a centered Gaussian vector with unit variances and $\E[XY] = \rho$.
Then
\[
\E\!\big[\mathrm{sgn}(X)\,\mathrm{sgn}(Y)\big]
=\frac{2}{\pi}\arcsin(\rho)=\frac{2}{\pi}\left(\rho+\frac{\rho^3}{6}+\frac{3\rho^5}{40}+O(\rho^7)\right).
\]
\end{lemma}
The reader can find a proof in \cite[Lemma 3.2]{goemans1995improved} for example, while the identity $\frac{2}{\pi}\arcsin(\rho)=\frac{2}{\pi}\left(\rho+\frac{\rho^3}{6}+\frac{3\rho^5}{40}+O(\rho^7)\right)$ is an immediate consequence of the Taylor expansion.
 We can now state and prove our main technical tool.

\begin{proposition}\label{prop:main_prop}
Let $B$ be a real $n\times n$ matrix and let $P$ be a rank-one orthogonal projection.
Assume that $\|(\mathrm{I}-P)B(\mathrm{I}-P)\|_{F}^2=n+\tilde{O}(1)$ and that the operator norms
$\|B\|$, and $\|B^{T}P\|$ are $\tilde{O}(1)$. Further assume $\Tr(B)=\tilde{O}(1)$.
Let $\bu,\hu$ denote the first two columns of a Haar-random orthogonal matrix in $\mathcal{O}(n)$. Then
\[
\E\!\left[\mathrm{sgn}\!\big(\langle P\bu,\hu\rangle\,
\langle B\bu,\hu\rangle\big)\right]
=\frac{2^{3/2}}{\pi^{3/2}}\Tr(PB)n^{-1/2}+\tilde{O}(n^{-3/2}).
\]
\end{proposition}
Before proving the proposition, we mention that for symmetric matrices, with no loss of generality one can assume that $B$ is a diagonal. In that case, we can take $P$ to be the orthogonal projection on a coordinate $e_k$, for some $k = 1,\dots,k$, which corresponds to an eigenspace of $B$ with, say, eigenvalue $\lambda_k$. With these choices, as long as $A$ is drawn from an orthogonally invariant ensemble, we shall show that
$$\E[\phi(A,k)] = \E\left[\mathrm{sgn}\!\big(\langle P\bu,\hu\rangle\,
\langle B\bu,\hu\rangle\big)\right],$$
where $B$ is the diagonalization of $A$. Thus Proposition \ref{prop:main_prop} provides a direct formula for the expected value of $\phi(A,k)$. Building in this representation and the ideas that will appear in the proof, we shall also use this formula to control the variance of $\phi(A,k)$.
\begin{proof}[Proof of Proposition \ref{prop:main_prop}]
Our proof is conducted in several steps. We first prove an analog result for quadratic forms in Gaussians. We then explain how to reduce, through conditioning, the computation on the orthogonal group to Gaussian space. We then estimate the error by integrating over the conditioned variables.

\emph{Step 0: Introduce the Gaussian setting.} Let $\bg$ and $\hg$ be two independent standard Gaussian vectors in $\R^n$. By rotational invariance of $(\bg,\hg)$, we may and will assume $\proj=e_{1}e_{1}^T$ is the projection onto the first coordinate. 
The well-known construction of the first two Haar columns is
\[
(\bu,\hu) \;\stackrel{\mathrm{law}}=\; \left(\frac{\bg}{\|\bg\|}, \frac{Q\,\hg}{\|Q\,\hg\|}\right),
\quad
Q=Q_{\bg}:=\mathrm{I}_n-\frac{\bg\bg^\top}{\|\bg\|^2}.
\]
The map
$(u,v)\mapsto \mathrm{sgn}\big(\langle Pu,v\rangle\,\langle Bu,v\rangle\big)$
is invariant under \emph{independent scalings} of $u$ and $v$.
Hence
\[
\mathrm{sgn}\!\big(\langle P\bu,\hu\rangle\,\langle B\bu,\hu\rangle\big)
\ \stackrel{\mathrm{law}}=\ 
\mathrm{sgn}\!\big(\langle P\bg,Q\hg\rangle\,\langle B\bg,Q\hg\rangle\big)=\ 
\mathrm{sgn}\!\big(\, g_{1}\,\langle Qe_1,\hg\rangle\,\langle QB\bg,\hg\rangle\big).
\]

\emph{Step 1: Condition on $\bg$ and reduce to two linear forms of $\hg$.}
Fix $\bg$ and set
\[
u:=\mathrm{sgn}(g_1)Qe_1,\qquad v:=QB\bg, \qquad
X:=\frac{\langle u,\hg\rangle}{\|u\|},\qquad
Y:=\frac{\langle v,\hg\rangle}{\|v\|}.
\]
For almost any value of $\bg$, $(X,Y)\mid\,\bg$ is bivariate normal, centered, with
\[
\Var(X\mid \bg)=1,\quad \Var(Y\mid \bg)=1,\quad
\rho:=\Cov(X,Y\mid \bg)=\frac{\langle u,v\rangle}{\|u\|\|v\|}=\mathrm{sgn}(g_1)\frac{\langle Qe_1,QB\bg\rangle}{\|Qe_1\|\|QB\bg\|},
\]
so that, by Lemma Lemma \ref{lem: signArcsine},
\[\E_{\hg}\big(\mathrm{sgn}\!\big(\langle P\bu,\hu\rangle\,\langle B\bu,\hu\rangle\big)\big)=\E_{\hg}\big(\mathrm{sgn}(X)\mathrm{sgn}(Y)\big)=\frac{2}{\pi}\left(\rho+\frac{\rho^3}{6}\right)+O(\rho^5).\]
It is now a matter of bounding the random correlation \(\rho\). Notice that $\langle Qe_1,QB\bg\rangle=\langle Qe_1,B\bg\rangle$ due to $Q^2=Q=Q^T$. The fact that $\|Qx\|^2=\|x\|^2-\frac{|\langle x,\bg\rangle|^2}{\|\bg\|^2}$ for any $x\in\R^n$ allows to write
\[
\mathrm{sgn}(g_{1})\rho
=\frac{\|\bg\|^2\langle Qe_1,B\bg\rangle}{\sqrt{(\|\bg\|^2-g_{1}^2)(\|\bg\|^2\|B\bg\|^2-|\langle \bg,B\bg\rangle|^2)}}
=\frac{\|\bg\|^2(B\bg)_1-g_{1}\langle \bg,B\bg\rangle}{\sqrt{(\|\bg\|^2-g_{1}^2)(\|\bg\|^2\|B\bg\|^2-|\langle \bg,B\bg\rangle|^2)}}.
\]
Define the fluctuations 
\begin{equation} \label{eq:errors}
   \eps:=\frac{\|\bg\|^2-n}{n},\quad \delta:=\frac{\|B\bg\|^2-n}{n},\quad \eta:=\frac{\langle\bg,B\bg\rangle}{n}.
\end{equation}
Notice that $\E[\eps]=0,\ \E[\delta]=n^{-1}(\|B\|_{F}^2-n),$ and $\E[\eta]=n^{-1}\Tr(B)$. The assumptions that $\|B\|_{F}^2-n=O(1)$ and $\Tr(B)=\tilde{O}(1)$, together with the Gaussian concentration from Lemma \ref{lm:typical}, show that with high probability (e.g. $1-n^{-\log n}$) the random variables $\eps,\delta,\eta$ are $\tilde{O}(n^{-1/2})$. Define the Gaussian 
\[Z=(B\bg)_1=\langle B^Te_1,\bg\rangle.\]
Notice that $\|B^Te_1\|=\|B^TPe_1\|\le\|B^TP\|=\tilde{O}(1)$ by assumption. So $g_1$ and $Z$ are centered Gaussians with variances $1$ and $\tilde{O}(1)$ respectively. In particular, with high probability $g_1$ and $Z$ are $\tilde{O}(1)$. We may conclude that with high probability (e.g. $1-n^{-\log n}$),
\begin{align*}
\mathrm{sgn}(g_{1})\rho=& \frac{n\,(1+\eps)Z-n\,\eta\, g_1}{\sqrt{\big(n\,(1+\eps)-g_1^2\big)\big(n^2(1+\eps)(1+\delta)-n^2\eta^2\big)}}\\
= &    n^{-1/2}\frac{Z-\eta g_1+\rem}{\sqrt{\big(1+\eps+g_{1}^2\,n^{-1}\big)\big(1+\eps+\delta+\eps\,\delta-\eta^2)\big)}}\\
= &    n^{-1/2}\frac{Z-\eta g_1+\rem}{\sqrt{1+2\eps+\delta+g_{1}^2\,n^{-1}+\eps\,\delta-\eta^2+\rem}}\\
= &    n^{-1/2}\big(Z-\eta g_1\big)\big(1-\frac{1}{2}(2\eps+\delta+g_{1}^2\,n^{-1}+\eps\,\delta-\eta^2)+\frac{3}{8}(2\eps+\delta)^2+\rem\big)\\
\\
= &    n^{-1/2}\big[Z\big(1-\eps-\frac{\delta}{2}-\frac{1}{2}(g_{1}^2\,n^{-1}+4\eps\,\delta-\eta^2+3\eps^2+\delta^{2}/4)\big)-\eta g_{1}\big(1-\eps-\frac{\delta}{2}\big)\big]+\tilde{O}(n^{-2})\\
= &    n^{-1/2}Z-n^{-1/2}\big(Z(\eps+\frac{\delta}{2})+\eta g_{1}\big)+W+\tilde{O}(n^{-2}),
\end{align*}
where $W$ is with high probability $\rem$
\[W:=n^{-1/2}\big(\eta g_1(\eps+\frac{\delta}{2})-Z\frac{1}{2}(g_{1}^2\,n^{-1}+4\eps\,\delta-\eta^2+3\eps^2+\delta^{2}/4)\big).\]
In particular, with probability at least  $\ 1-n^{-\log n}$,
\[\rho=\tilde{O}(n^{-1/2}),\quad\text{and} \quad\rho=n^{-1/2}\mathrm{sgn}(g_{1})\left(Z-Z\eps-Z\frac{\delta}{2}+\eta g_1\right)+\rem. \]
Since $|\mathrm{sgn}| \leq 1$, we can now calculate up to $\rem$,
\begin{align*}
    \E\!\left[\mathrm{sgn}\!\big(\langle P\bu,\hu\rangle\,
\langle B\bu,\hu\rangle\big)\right]
= & \frac{2}{\pi}\E_{\bg}[\rho]+\rem.
\end{align*}
\emph{Step 2: Integration over $\bw=(\mathrm{I}_{n}-e_{1}e_{1}^{T})\bg$.}
To do so, we let $\bw=(0,g_{2},\ldots,g_{n})$, so that its $n-1$ non-zero entries form a standard $n-1$ dimensional Gaussian independent of $g_1$, and we first integrate over $\bw$. We can write 
\[\eps=\frac{\|\bw\|^2+g_{1}^2-n}{n},\quad \delta:=\frac{\|B\bw\|^2+g_{1}B_{11}-n}{n},\quad \eta:=\frac{B_{11}g_{1}^2+\langle\bw,B\bw\rangle+\langle\bw,Be_{1}\rangle+\langle e_{1},B\bw\rangle}{n},\]
and
\[Z=B_{11}g_{1}+\langle B^{T}e_{1},\bw\rangle.\]
Notice that by our assumptions, $\E_{\bw}\|B\bw\|^2=\|(\mathrm{I}-P)B(\mathrm{I}-P)\|_{F}^2=n+\tilde{O}(1)$ and $\E_{\bw}\langle \bw,B\bw\rangle=\Tr((\mathrm{I}-P)B(\mathrm{I}-P))=\Tr(B)-\Tr(PB)=\tilde{O}(1)$. Using these observations and the symmetry $\bw\mapsto-\bw$, we get  
\begin{align*}
\E_{\bw}[Z]=& B_{11}g_{1},\\
\E_{\bw}[Z\eps]=& \E_{\bw}[B_{11}g_{1}\eps]=\frac{B_{11}g_{1}(g_{1}^2-1)}{n}=\tilde{O}(1/n),\\
\E_{\bw}[Z\delta]=& \E_{\bw}[B_{11}g_{1}\delta]=\frac{B_{11}g_{1}(B_{11}g_{1}+\tilde{O}(1))}{n}=\tilde{O}(1/n),\\
\E_{\bw}[\eta]=& \frac{B_{11}g_{1}^2+\tilde{O}(1)}{n}=\tilde{O}(1/n).
\end{align*}
We conclude that,
\begin{align*}
\E_{\bg}[\rho]=\E_{g_{1}}\left[\E_{\bw}(\rho)\right]= & n^{-1/2}\Tr(PB)\E_{g_{1}}\left[|g_{1}|\right]+\rem\\ = & \sqrt{\frac{2}{\pi}}n^{-1/2}\Tr(PB)+\rem,  
\end{align*}
and therefore
\begin{equation*}
 \E\!\left[\mathrm{sgn}\!\big(\langle P\bu,\hu\rangle\,
\langle B\bu,\hu\rangle\big)\right]
=  \frac{2}{\pi}\E_{\bg}[\rho]+\rem
= \frac{2^{3/2}}{\pi^{3/2}}\Tr(PB)n^{-1/2}+\rem. 
\end{equation*}
This proves the claim.
\end{proof}

\section{Statistics of the Nodal Count}
In this section, we employ Proposition \ref{prop:main_prop} and prove Theorem \ref{thm:main}. Recall the setting:  $A=\Phi\mathrm{diag(\bm\Lambda)\Phi^T}$, where $\Phi \sim \mathrm{Haar}(\mathcal{O}(n))$, and $\bm{\Lambda}=(\lambda_{1},\ldots,\lambda_{n})$ can be random, but $\Phi$ and $\bm{\Lambda}=(\lambda_{1},\ldots,\lambda_{n})$ are independent. We can and will assume without loss of generality that $\av(\bm\Lambda)=0$ and $\std(\bm\Lambda)=1$ with probability $1$. We also assume the spectral growth bound from Definition \ref{spec_growth}, according to which $\|\bm\Lambda\|_{\infty}=\tilde{O}(1)$ with probability $1-\rem$. Our goal is to prove Theorem \ref{thm:main}, which states that 
\begin{align}
    \E_{\Phi,\bm\Lambda}[\phi(A,k)]= & { n \choose 2}\left(\frac{1}{2}+\frac{\sqrt{2}}{\pi^{3/2}}\E_{\bm\Lambda}[\lambda_{k}]n^{-1/2}+\rem\right),\quad\text{and}\\
    \Var_{\Phi,\bm\Lambda}[\phi(A,k)]= &\ \tilde{O}(n^{5/2}).
\end{align}

The first part of the proof is a simple reduction based on the permutation symmetry of rows in $\mathcal{O}(n)$. Let $\Phi_{j}=(\Phi_{j,1},\ldots,\Phi_{j,n})$ denote the $j$-th row of $\Phi$, let $P_{k}=e_{k}e_{k}^{T}$ be the projection on the $k$-th coordinate, and for $i<j$ and $r<s$ define the random variable
\[M_{ij}:=A_{ij}\Phi_{i,k}\Phi_{j,k}=\langle\Phi_i,P_{k}\Phi_{j}\rangle\,\langle\Phi_{i},\diag(\bm\Lambda)\Phi_j\rangle,\]
so that
\[\phi(A,k)=\frac{1}{2}\sum_{i<j}(1+\mathrm{sign}(M_{ij})).\]
Notice that the permutation symmetry of rows in $\mathcal{O}(n)$ gives  
\begin{align*}
     M_{ij}\stackrel{\mathrm{law}}{=} & M_{12},\\
     M_{ij}M_{rs}\stackrel{\mathrm{law}}{=} & \begin{cases}M_{12}^2 &\ (i,j)=(r,s)\\
     M_{12}M_{34} &\ i\ne r,j\ne s\\
     M_{13}M_{23} &\ i\ne r,j= s\quad  \text{or}\quad i= r,j\ne s\end{cases}.
\end{align*}
Linearity of expectation and bilinearity of covariance allow us to conclude
\begin{align*}
    \E[\phi(A,k)]= &\ \frac{1}{2}{n \choose 2}\left(1+\E[\mathrm{sign}(M_{12})]\right),\quad\text{and}\\
    \Var(\phi(A,k))= &\ \Cov\left(\frac{1}{2}\sum_{i<j}\mathrm{sign}(M_{ij}),\ \frac{1}{2}\sum_{r<s}\mathrm{sign}(M_{rs})\right)\\
    = &\  \frac{1}{4} {n \choose 2} \Var( \mathrm{sgn}\left(M_{12} \right))+ \frac{1}{2} (n-2) {n \choose 2} \Cov\left( \mathrm{sgn}(M_{13}),\mathrm{sgn}(M_{23}))\right)\\
    & +2 {n \choose 2} {n-2 \choose 2}\Cov(\mathrm{sgn}\left(M_{12} M_{34} \right)).
\end{align*}
We now prove Theorem \ref{thm:main} based on three Lemmas. The second part of the proof, which is the technical part, is to prove these Lemmas. 
\begin{lemma}[Expected value for an edge]\label{lm:expectedvalue} Let $\bm\Lambda$ be fixed, and assume it satisfies $\av(\bm\Lambda)=0$, $\std(\bm\Lambda)=1,$ and $\|\bm\Lambda\|_{\infty}=\tilde{O}(1)$. Then, 
$$ \E_{\Phi}\left[ \mathrm{sgn}\left(M_{12} \right)\right] = \frac{2^{3/2}}{\pi^{3/2}}\lambda_{k}n^{-1/2}+\rem.$$
\end{lemma}
This lemma is straightforward from Proposition \ref{prop:main_prop}.
\begin{proof}
    The assumption on $\bm\Lambda$ implies that $B=\diag(\bm\Lambda)$ and $P=e_{k}e_{k}^{T}$ satisfy the conditions of Proposition \ref{prop:main_prop}, and $\Tr(BP)=\lambda_{k}$. which tells us that
    $$\E_{\Phi}\left[ \mathrm{sgn}\left(M_{12} \right)\right]=\E_{\Phi}\left[ \mathrm{sgn}(\langle\Phi_{1},P\Phi_2)\rangle\langle\Phi_{1},\diag(\bm\Lambda)\Phi_2)\rangle\right]=\frac{2^{3/2}}{\pi^{3/2}}n^{-1/2}\lambda_k+\rem.$$
\end{proof}
\begin{lemma}[Covariance of Adjacent Edges]\label{lm:sharedvertex}
Let $\bm\Lambda$ be fixed, and assume it satisfies $\av(\bm\Lambda)=0$, $\std(\bm\Lambda)=1,$ and $\|\bm\Lambda\|_{\infty}=\tilde{O}(1)$. Then,  
$$ \Cov_{\Phi}\left( \mathrm{sgn}(M_{13}),\mathrm{sgn}(M_{23}))\right) = \tilde O(n^{-1/2}).$$
\end{lemma}
\begin{lemma}[Covariance of Non-Adjacent Edges]\label{lm:separate_edges}
Let $\bm\Lambda$ be fixed, and assume it satisfies $\av(\bm\Lambda)=0$, $\std(\bm\Lambda)=1,$ and $\|\bm\Lambda\|_{\infty}=\tilde{O}(1)$. Then,  
$$ \Cov_{\Phi}\left( \mathrm{sgn}(M_{12}),\mathrm{sgn}(M_{34}))\right) = \rem$$
\end{lemma}
The proofs of Lemma \ref{lm:separate_edges} and Lemma \ref{lm:sharedvertex} are in Subsection \ref{sec:variance}. We now use these lemmas to prove Theorem \ref{thm:main}.

\begin{proof}[Proof of Theorem \ref{thm:main} based on Lemmas $\ref{lm:expectedvalue},\ref{lm:sharedvertex},\ref{lm:separate_edges}$. ]
    Let $\mathcal{E}_{T}$ be the event of $\|\bm\Lambda\|_{\infty}=\tilde{O}(1)$ and let $\on{T}$ be the indicator of $\mathcal{E}_{T}$. By assumption, $\E(1-\on{T})=\rem$, so for any $x,y$ random variables
    \[\E[\mathrm{sign}(x)]= \E[\mathrm{sign}(x)\on{T}+\mathrm{sign}(x)(1-\on{T})]=\E[\mathrm{sign}(x)\on{T}]+\rem,\]
    and similarly, 
    \[\Cov\left(\mathrm{sign}(x),\mathrm{sign}(y)\right)=\Cov\left(\mathrm{sign}(x)\on{T},\mathrm{sign}(y)\right)+\rem.\]
    By Lemma \ref{lm:expectedvalue}, we conclude the needed expectation result
$$      \E[\phi(A,k)]=\frac{1}{2}{n \choose 2}\left(1+\E[\mathrm{sign}(M_{12}]\on{T})+\rem\right)={n \choose 2}\left(\frac{1}{2}+\frac{\sqrt{2}}{\pi^{3/2}}n^{-1/2}\E_{\bm\Lambda}[\lambda_{k}]+\rem\right).  
 $$ 
Using that $\Var(\mathrm{sign}(M_{12}))\le1$ by definition, together with Lemma \ref{lm:separate_edges} and Lemma \ref{lm:sharedvertex}, we get the needed variance bound
\begin{align*}
    \Var(\phi(A,k))\ \le &\  \frac{1}{4} {n \choose 2} + \frac{1}{2} (n-2) {n \choose 2}\left( \Cov\left( \mathrm{sgn}(M_{13})\on{T},\mathrm{sgn}(M_{23}))\right)+\rem\right)\\
    &\qquad+2 {n \choose 2} {n-2 \choose 2}\left( \Cov\left( \mathrm{sgn}(M_{12})\on{T},\mathrm{sgn}(M_{34}))\right)+\rem\right)\\
    = &\ \tilde{O}(n^{5/2}).
\end{align*}
\end{proof}

\subsection{Estimating $\Cov(\mathrm{sign}(M_{ij}),\mathrm{sign}(M_{rs}))$ - Proofs of Lemma \ref{lm:sharedvertex} and Lemma \ref{lm:separate_edges}}\label{sec:variance}

For convenience denote $(\bu,\hu,\bv,\hv)=(\Phi_{1},\Phi_2,\Phi_3,\Phi_4)$, the four columns of a Haar random orthogonal matrix $\Phi$.\footnote{Equivalently, $(\bu,\hu,\bv,\hv)$ is a random element of the Stiefel manifold $V_{4} (\R^n)$ with the the uniform probability measure.} 
We assume $\bm\Lambda=(\lambda_1,\ldots,\lambda_n)$ is fixed with $\av(\bm\Lambda)=0,\ \std(\bm\Lambda)=1$, and $\|\bm\Lambda\|_{\infty}=\tilde{O}(1)$. Define the quartic polynomial $M:\R^n\times \R^n\to \R$ by
\begin{equation} \label{eq:Mdef}
    M(x,y) := x_ky_k \sum\limits_{j=1}^n\lambda_jx_jy_j=\langle e_{k}e_{k}^Tx,y\rangle\langle \diag(\bm\Lambda)x,y\rangle,
\end{equation}
so that
\begin{equation} \label{eq:identities}
    M_{12} = M(\bu,\hu), \quad M_{13} = M(\bu,\bv),\quad \quad M_{23} = M(\hu,\bv),\quad M_{34} = M(\hu,\hv).
\end{equation}

To help with the conditioning, we shall reduce some of the variables to Gaussian space. For that, we also introduce the (random) isometry $U:\R^{n-2} \to \mathrm{span}(\bu,\hu)^\perp \subset \R^{n}.$ In other words, $U$ is the unique, up to a change of basis, matrix such that $U^TU = \mathrm{I}_{n-2}$ and $UU^T = \mathrm{I}_n - \bu\bu^T - \hu\hu^T.$




\begin{proof}[Proof of Lemma \ref{lm:sharedvertex}]
Let $\bg\sim N(0,\mathrm{I}_{n-2})$ be a standard Gaussian vector independent of $\bu,\hu$. By the definition of the matrix $U$, we have 
\begin{align*}
    \bv\ \mid\ (\bu,\hu) \stackrel{\mathrm{law}}{=}&\ \frac{U\bg}{\|U\bg\|}\\
    \mathrm{sgn}(M_{13}M_{23})\ \mid\ (\bu,\hu)
  \stackrel{\mathrm{law}}{=}&\ \mathrm{sgn}(M(\bu,U\bg)M(\hu,U\bg))\\
 =&\ \mathrm{sgn}\big(\langle e_{k}e_{k}^{T}\bu,U\bg\rangle\langle\diag(\bm\Lambda)\bu,U\bg\rangle\,\langle e_{k}e_{k}^{T}\hu,U\bg\rangle\langle\diag(\bm\Lambda)\hu,U\bg\rangle\big)
 \\
 =&\ \mathrm{sgn}(u_k\,\hat{u}_k)\mathrm{sgn}\big(\langle U^{T}e_{k},\bg\rangle\langle U^T\diag(\bm\Lambda)\bu,\bg\rangle\,\langle U^Te_{k},\bg\rangle\langle U^T\diag(\bm\Lambda)\hu,\bg\rangle\big)
 \\
 =&\ \mathrm{sgn}(u_k\,\hat{u}_k)\mathrm{sgn}\big(\langle U^T\diag(\bm\Lambda)\bu,\bg\rangle\,\langle U^T\diag(\bm\Lambda)\hu,\bg\rangle\big)
\end{align*}
Denote the following $(\bu,\hu)$-dependent parameters $$\alpha_{\bu}:=\|\diag(\bm\Lambda)\bu\|,\ \alpha_{\hu}:=\|\diag(\bm\Lambda)\hu\|,\ \beta_{\bu}=\langle\bu,\diag(\bm\Lambda)\bu\rangle,\ \beta_{\hu}=\langle\hu,\diag(\bm\Lambda)\hu\rangle,\ \gamma:=\langle\bu,\diag(\bm\Lambda)\hu\rangle,$$ and define \begin{align*}
\rho:= &\ \frac{\langle U^T\diag(\bm\Lambda)\bu,U^T\diag(\bm\Lambda)\hu\rangle}{\|U^T\diag(\bm\Lambda)\bu\|\|U^T\diag(\bm\Lambda)\hu\|}\\= &\ \frac{\langle \diag(\bm\Lambda)UU^T\diag(\bm\Lambda)\bu,\hu\rangle}{\|UU^T\diag(\bm\Lambda)\bu\|\|UU^T\diag(\bm\Lambda)\hu\|}\\= &\ \frac{\langle \diag(\bm\Lambda)^2\bu,\hu\rangle-\beta_{\bu}\gamma-\beta_{\hu}\gamma}{\sqrt{(\alpha_{\bu}^2-\beta_{\bu}^2-\gamma^2)(\alpha_{\hu}^2-\beta_{\hu}^2-\gamma^2)})},
\end{align*} so that Lemma \ref{lem: signArcsine} gives
\[\E\left[\mathrm{sgn}(M_{13}M_{23})\ \mid\ (\bu,\hu)\right]=\frac{2}{\pi}\mathrm{sgn}(u_k\,\hat{u}_k)\arcsin{\rho}=\frac{2}{\pi}\mathrm{sgn}(u_k\,\hat{u}_k)\rho+O(\rho^3).\]
Lemma \ref{lm:typical} allows to estimate the above parameters. First consider $\xi,\hat{\xi}$ independent standard Gaussian vectors of size $n$, and let $Q=\mathrm{I}_{n}-\frac{\xi\,\xi^{T}}{\|\xi\|^2}$, so that 
\[(\bu,\hu) \stackrel{\mathrm{law}}{=}\left(\frac{\xi}{\|\xi\|},\frac{Q\hat{\xi}}{\|Q\hat{\xi}\|}\right).\]
Denote the following $(\xi,\hat{\xi})$-dependent small parameters $$\eps_{1}:=\frac{\|\xi\|^2-n}{n},\ \eps_{2}:=\frac{\|\hat{\xi}\|^2-n}{n},\ \eps_{3}:=\frac{\|\diag(\bm\Lambda)\xi\|^2-n}{n},\ \eps_{4}:=\frac{\|\diag(\bm\Lambda)\hat{\xi}\|^2-n}{n},$$
$$\eps_{5}=\frac{\langle\xi,\diag(\bm\Lambda)\xi\rangle}{n},\ \eps_{6}=\frac{\langle\hat{\xi},\diag(\bm\Lambda)\hat{\xi}\rangle}{n},\ \eps_{7}:=\frac{\langle\xi,\diag(\bm\Lambda)\hat{\xi}\rangle}{n},\ \eps_{8}:=\frac{\langle\xi,\diag(\bm\Lambda)^2\hat{\xi}\rangle}{n}.$$
Notice that $\E[\eps_{j}]=0$ for all $j$ because of the $\xi,\hat{\xi}$ independence and the normalization $\av(\bm\Lambda)=0,\ \std(\bm\Lambda)=1$. Moreover, each $\eps_j$ is a quadratic form in the variables $\xi,\hat{\xi}$ and so, together with the bound on $\|\bm\Lambda\|_{\infty}$, the Gaussian concentration from Lemma \ref{lm:typical} applies, as in \eqref{eq:specconc}. This means that the typical event $\mathcal{E}_{T}$ for which $\eps_{j}=\tilde{O}(n^{-1/2})$ for all $j$ uniformly has $\P(\mathcal{E}_{T})\ge 1-n^{n\log n}$ by Lemma \ref{lm:typical}. We can write
$$\alpha_{\bu}^2=\frac{1+\eps_3}{1+\eps_1},\  \alpha_{\hu}^2=\frac{1+\eps_4}{1+\eps_2},\ \beta_{\bu}=\frac{\eps_{5}}{(1+\eps_1)},\ \beta_{\hu}=\frac{\eps_{6}}{(1+\eps_2)},\ \gamma:=\frac{\eps_{7}}{\sqrt{(1+\eps_1)(1+\eps_2)}},$$
which means that under the typical event 
$$\rho\, \on{T}=\frac{\frac{\eps_8}{\sqrt{(1+\eps_1)((1+\eps_2))}}-\frac{\eps_5\,\eps_7}{(1+\eps_1)\sqrt{(1+\eps_1)((1+\eps_2))}}-\frac{\eps_6\,\eps_7}{(1+\eps_2)\sqrt{(1+\eps_1)((1+\eps_2))}}}{\sqrt{\frac{1+\eps_3}{1+\eps_1}-\frac{\eps_5^2}{(1+\eps_1)^2}-\frac{\eps_7^2}{(1+\eps_1)(1+\eps_2)}}\sqrt{\frac{1+\eps_4}{1+\eps_2}-\frac{\eps_6^2}{(1+\eps_2)^2}-\frac{\eps_7^2}{(1+\eps_1)(1+\eps_2)}}}\, \on{T}=\tilde{O}(n^{-1/2}).$$
Therefore
\[\E\left[\mathrm{sgn}(M_{13}M_{23})\right]=\E_{\bu,\hu}\left[\frac{2}{\pi}\mathrm{sgn}(u_k\,\hat{u}_k)\arcsin{\rho}\on{T}\right]+\tilde{O}(n^{-\log n})=\tilde{O}(n^{-1/2}).\]
Since $\E\left[\mathrm{sgn}(M_{13})\right]\E\left[\mathrm{sgn}(M_{23})\right]=\E\left[\mathrm{sgn}(M_{12})\right]^2=\tilde{O}(n^{-1})$ by Lemma \ref{lm:expectedvalue}, we conclude that $$\Cov(\mathrm{sgn}(M_{13}),\mathrm{sgn}(M_{23}))=\tilde{O}(n^{-1/2}).$$
\end{proof}
\begin{proof}[Proof of Lemma \ref{lm:separate_edges}]
Our desired result is $I=\tilde{O}(n^{-3/2})$, where
\[I:=\E[\mathrm{sgn}\left(M(\bu,\hu)M(\bv,\hv)\right)]-\E[\mathrm{sgn}(M(\bu,\hu))]\E[\mathrm{sgn}\left(M(\bu,\hu)\right)].\]
Notice that, by \eqref{eqn:expectation}, 
\begin{align*}
I&=\E_{\bu,\hu}\left[\mathrm{sgn}\left(M(\bu,\hu)\right)\E_{\bv,\hv}\left[\mathrm{sgn}\left(M(\bv,\hv)\right)-\mathrm{sgn}\left(M(\bu,\hu)\right)\right]\right]\\
&=\E_{\bu,\hu}\left[\mathrm{sgn}\left(M(\bu,\hu)\right)\E_{\bv,\hv}\left[\mathrm{sgn}\left(M(\bv,\hv)\right)-\frac{2^{3/2}}{\pi^{3/2}}n^{-1/2}\lambda_{k}\right]\right]+\rem.
\end{align*}
For a fixed choice of $\bu,\hu$, assuming only that $(\bu)_k,(\hu)_k$ are $\tilde{O}(n^{-1/2})$, we recall the isometry $U:\R^{n-2}\to\mathrm{span}(\bu,\hu)^{\perp}$ as above.
Let $\xi,\hat{\xi}$ be the first two columns of a Haar-random orthogonal $(n-2)\times (n-2)$ matrix. Define $P=\frac{1}{\|U^{T}e_k\|^2}U^Te_{k}e_{k}^{T}U$ and $B=U^T\diag(\bm\Lambda)U$. Then, 
\begin{align*}
    (\bv,\hv)\ \mid\ (\bu,\hu)& \stackrel{\mathrm{law}}{=}(U\xi,U\hat{\xi})\\
    \mathrm{sgn}(M(\bv,\hv))\ \mid\ (\bu,\hu)
 & \stackrel{\mathrm{law}}{=}\mathrm{sgn}(\langle e_{k}e_{k}^{T}U\xi,U\hat{\xi}\rangle\langle\diag(\bm\Lambda)U\xi,U\hat{\xi}\rangle)\\
 & = \mathrm{sgn}(\langle P\xi,\hat{\xi}\rangle\langle B\xi,\hat{\xi}\rangle)
\end{align*}
If we can show that $P$ and $B$ satisfy the assumptions of Proposition \ref{prop:main_prop}, and that $\Tr(PB)=\lambda_k+\tilde{O}(1/n)$, then we get that with high probability in $\bu,\hu$,
\[\E_{\bv,\hv}\left[\mathrm{sgn}\left(M(\bv,\hv)\right)\right]-\frac{2^{3/2}}{\pi^{3/2}}n^{-1/2}\lambda_{k}=\rem,\]
and therefore,
\[I=\E_{\bu,\hu}\left[\mathrm{sgn}\left(M(\bu,\hu)\right)\left[\rem\right]\right]+\rem=\rem.\]
We are left with showing that $P$ and $B$ satisfy the assumptions of Proposition \ref{prop:main_prop}, and that $\Tr(PB)=\lambda_k+\tilde{O}(1/n)$. Notice that the operator norm of $B$ is bounded by $\|B\|\le \|\bm\Lambda\|_{\infty}=\tilde{O}(1)$ by assumption. Similarly, the assumptions $\|\bm\Lambda\|_{\infty}=\tilde{O}(1)$ and $\av(\bm\Lambda)=0$ gives $\Tr(B)=\Tr(\diag(\bm\Lambda))-\langle\bu,\diag(\bm\Lambda)\bu\rangle-\langle\hu,\diag(\bm\Lambda)\hu\rangle=\tilde{O}(1)$. The assumption that $ \std(\bm\Lambda)=1$ allows to bound the Frobenius norm: 
$\left|\|B\|_{F}^2-\|\bm\Lambda\|_{2}^2\right|\le 2\|\bm\Lambda\|_{\infty}^2$ means $\|B\|_{F}^2=n+\tilde{O}(1)$. Similarly, $\left|\|(\mathrm{I}_{n-2}-P)B(\mathrm{I}_{n-2}-P)\|_{F}^2-\|B\|_{F}^2\right|\le 2\|\bm\Lambda\|_{\infty}^2$ so $\|(\mathrm{I}_{n-2}-P)B(\mathrm{I}_{n-2}-P)\|_{F}^2=n+\tilde{O}(1)$. Finally, $\|B^{T}P\|\le\|B\|=\tilde{O}(1)$. We conclude that $P$ and $B$ satisfy the assumptions of Proposition \ref{prop:main_prop}. We now calculate 
\begin{align*}
 \Tr(PB)= & \frac{1}{\|U^{T}e_k\|^2}\Tr(U^Te_{k}e_{k}^{T}UU^T\diag(\bm\Lambda)U)\\
    = & \frac{1}{1-(\bu)_k^2-(\hu)_k^2}\Tr(e_{k}e_{k}^{T}(\mathrm{I}_n-\bu\bu^T-\hu\hu^{T})\diag(\bm\Lambda)(\mathrm{I}_n-\bu\bu^T-\hu\hu^{T}))\\
    = & \frac{
\lambda_{k}-2\lambda_k((\bu)_k^2+(\hu)_k^2)+\langle\diag(\bm\Lambda)(\bu\bu^T+\hu\hu^{T})e_{k},(\bu\bu^T+\hu\hu^{T})e_{k}\rangle
    }{1-(\bu)_k^2-(\hu)_k^2}
\end{align*}
so
\begin{align*}
    \left|\Tr(PB)-\lambda_k\right|= & \left|\frac{
-\lambda_k((\bu)_k^2+(\hu)_k^2)+\langle\diag(\bm\Lambda)(\bu\bu^T+\hu\hu^{T})e_{k},(\bu\bu^T+\hu\hu^{T})e_{k}\rangle
    }{1-(\bu)_k^2-(\hu)_k^2}\right|\\   
     \le& 2\|\bm\Lambda\|_{\infty}\left|\frac{
(\bu)_k^2+(\hu)_k^2}{1-(\bu)_k^2-(\hu)_k^2}\right|=\tilde{O}(1/n).
\end{align*}
\end{proof}
\section*{Acknowledgements}
The authors thank the Department of Mathematics at the Massachusetts Institute of Technology, where this project was initiated.
L.A.\ was supported by the Department of Mathematics at MIT and by the Simons Foundation Grant No.~601948 (D.J.).
D.M.\ was partially supported by the Brian and Tiffinie Pang Faculty Fellowship.
This material is based upon work supported by the National Science Foundation under Grant No.~DMS-2513687.
J.U.\ thanks Mehtaab Sawhney for many stimulating conversations on the subject.
The authors are also grateful to Louisa Thomas for helpful suggestions that improved the presentation.

{ \small 
	\bibliographystyle{plain}
	\bibliography{main.bib} }

\clearpage
\phantomsection

\end{document}